\documentclass[11pt,a4paper,titling,fleqn]{article}
\usepackage[paper=a4paper,left=25mm,right=25mm,top=25mm,bottom=25mm]{geometry}
\emergencystretch=\maxdimen
\usepackage{amsthm,amsmath,amsfonts,amssymb}
\usepackage[square,numbers]{natbib}
\usepackage{graphicx,psfrag,epsf}
\usepackage{enumerate}
\usepackage{url}
\usepackage{hyperref} 
\hypersetup{colorlinks=true,allcolors=[rgb]{0,0,1}}
\usepackage{color,xcolor}
\usepackage{dsfont}
\usepackage{multirow}
\usepackage{array}
\usepackage{booktabs}

\allowdisplaybreaks

\selectfont
\definecolor{red}{rgb}{0.7, 0.11, 0.11}

\newcommand{\Mail}[1]{{\color{red} #1}}

\begin{document}
\theoremstyle{plain}
\newtheorem{theorem}{Theorem}[section]
\newtheorem{proposition}[theorem]{Proposition}
\newtheorem{corollary}[theorem]{Corollary}
\newtheorem{lemma}[theorem]{Lemma}
\theoremstyle{remark}
\newtheorem{remark}[theorem]{Remark}
\newtheorem{definition}[theorem]{Definition}
\newtheorem{assumption}[theorem]{Assumption}
\newtheorem{example}[theorem]{Example}
\newcounter{eqnum}
\setcounter{eqnum}{0}

\newcommand{\nexteq}{\stepcounter{eqnum}\tag{\theeqnum}}
\newcommand{\de}{\mathrm{\,d}}

\title{\bf On exact regions between measures of concordance \\ and Chatterjee's rank correlation  \\ for lower semilinear copulas}

\author{%
	Sebastian Fuchs\footnote{University of Salzburg, Austria, Email: \Mail{sebastian.fuchs@plus.ac.at}}
	  \qquad
    Carsten Limbach\footnote{University of Salzburg, Austria, Email: \Mail{carsten.limbach@plus.ac.at}}
    \qquad
	Fabian Schürrer\footnote{University of Salzburg, Austria, Email: \Mail{fabian.schuerrer@stud.plus.ac.at}}
	\vspace{5mm}
}

\maketitle

\begin{abstract}
We explore how the classical concordance measures--Kendall's $\tau$, Spearman's rank correlation $\rho$, and Spearman's footrule $\phi$--relate to Chatterjee’s rank correlation $\xi$ when restricted to lower semilinear copulas. 
First, we provide a complete characterization of the attainable $\tau$-$\rho$ region for this class, thus resolving the conjecture of \citet{maislinger2025}.
Building on this result, we then derive the exact $\tau$-$\phi$ and $\phi$-$\rho$ regions, obtain a closed-form relationship between $\xi$ and $\tau$, and establish the exact $\tau$-$\xi$ region.
In particular, we prove that $\xi$ never exceeds $\tau$, $\rho$, or $\phi$.
Our results clarify the relationship between undirected and directed dependence measures and reveal novel insights into the dependence structures that result from lower semilinear copulas.
\end{abstract}



\section{Introduction}

Dependence measures play a central role in quantifying the various types of dependence relationships between random variables. In practice, one typically models these relationships using copula families--such as Gaussian, Archimedean or extreme‐value copulas--that admit closed-form expressions for the desired dependence measures.

Another particularly tractable family of copulas is given by the lower (or, symmetrically, upper) semilinear copulas.
Geometrically, one begins by splitting the unit square $[0,1]^2$ along its main diagonal and then imposing linearity on each side, while preserving all defining properties of a copula; see the introductory paper \cite{durante2008} for an illustration and a probabilistic interpretation.
Concretely, let $\delta: [0,1] \to [0,1]$ be a \emph{copula diagonal}, i.e.~a non-decreasing, 2-Lipschitz function such that $\delta(t) \leq t$ for all $t \in [0,1]$, satisfying the shape constraints 
\begin{align}\label{Def:LSLDiagonal}
  t \mapsto \frac{\delta(t)}{t} \textrm{ is non-decreasing }
  \qquad \textrm{ and } \qquad
  t \mapsto \frac{\delta(t)}{t^2} \textrm{ is non-increasing}\,,
\end{align} 
implying that $t^2 \leq \delta(t) \leq t$ for all $t \in [0,1]$.
We denote by \(\mathcal{D}^{\rm LSL}\) the convex class of all those functions $\delta$.
Then, for such a \(\delta \in \mathcal{D}^{\rm LSL}\) the associated \emph{lower semilinear copula} $S_\delta: [0,1]^2 \to [0,1]$ is defined by \cite{durante2008}
\begin{align}\label{Def:LSL}
    S_\delta(u,v)
    & := 
    \begin{cases}
      v \; \frac{\delta(u)}{u} & \textrm{ if } v \leq u
      \\
      u \; \frac{\delta(v)}{v} & \textrm{ otherwise. }
    \end{cases}
\end{align}
The class of lower semilinear copulas is convex, compact under the sup-norm \cite{maislinger2025} and denoted by \(\mathcal{C}^{\rm LSL}\).
According to \cite{durante2008}, if $C \in \mathcal{C}^{\rm LSL}$ with copula diagonal $\delta(t) = C(t,t)$ for all $t \in [0,1]$, then $\delta \in \mathcal{D}^{\rm LSL}$.
The converse direction does not hold, in general; see Example \ref{ex:MO}.
This semilinear construction admits closed-form expressions for key dependence measures such as Kendall's tau $\tau$, Spearman's rho $\rho$ and Spearman's footrule $\phi$ as follows \cite{durante2006}: 
\begin{align}\label{LSL:MoC}
    \tau(S_\delta)
    & := 4 \, \int_{[0,1]} \frac{\delta^2(t)}{t} \de \lambda(t) - 1\,, 
    \\
    \rho(S_\delta)
    & := 12 \, \int_{[0,1]} t \, \delta(t) \de \lambda(t) - 3\,, \notag
    \\
    \phi(S_\delta)
    & := 6 \, \int_{[0,1]} \delta(t) \de \lambda(t) - 2\,. \notag
\end{align}
For continuous random variables $X$ and $Y$ whose joint distribution function is distributed according to a copula $S_\delta$, concordance measures evaluate the extent of monotone dependence between $X$ and $Y$. Such a functional is $1$ ($-1$) if $X$ and $Y$ are comonotone (countermonotone) and $0$ if $X$ and $Y$ are independent; see \cite{nelsen2006}.

In \cite{maislinger2025}, the authors investigated the $\tau$-$\rho$ region $\Omega_{\tau,\rho}^{\rm LSL}$ for lower semilinear copulas, namely 
\begin{align*}
  \Omega_{\tau,\rho}^{\rm LSL} := \{(\tau(S_\delta),\rho(S_\delta)): S_\delta \in \mathcal{C}^{\rm LSL}\}\,,
\end{align*}
the set of all attainable values of the pair $(\tau(S_\delta),\rho(S_\delta))$ with $S_\delta \in \mathcal{C}^{\rm LSL}$. 
The authors conjectured that 
\begin{align}\label{Conjecture}
  \Omega_{\tau,\rho}^{\rm LSL}
  = \{(x,y) \in [0,1]^2: x \leq y \leq 1 - (1-x)^{3/2}\}\,,
\end{align}
but to date have only succeeded in proving the lower inequality $\tau(S_\delta) \leq \rho(S_\delta)$.

In this paper, we present a proof of the upper inequality $\rho(S_\delta) \leq 1-(1-\tau(S_\delta))^{3/2}$ using a Hölder-type argument, thereby confirming the conjecture in \eqref{Conjecture}.
We further offer a streamlined proof of the lower inequality that avoids the use of Markov kernels.
Then, leveraging these techniques, we determine both the $\rho$-$\phi$ and $\tau$-$\phi$ region for lower semilinear copulas, demonstrating that the resulting concordance values for any given lower semilinear copula differ only slightly.
The latter result is not too surprising, since concordance measures capture essentially quite similar types of dependence relationships.

We then extend our study and relate the above concordance measure also to Chatterjee's rank correlation introduced in \cite{chatterjee2020} and given, for random variables $X$ and $Y$ whose joint distribution function is associated with copula $C$, by
\begin{align}\label{Def:Chatterjee}
  \xi(C)
  & := \frac{\int_{\mathbb{R}} {\rm var} (P(Y \geq y \, | \, X)) \; \mathrm{d} P^{Y}(y)}
					{\int_{\mathbb{R}} {\rm var} (\mathds{1}_{\{Y \geq y\}}) \; \mathrm{d} P^{Y}(y)}\,.
\end{align}
$\xi$ is also known as Dette-Siburg-Stoimenov measure \cite{siburg2013} and quantifies the extent of functional dependence of $Y$ on $X$. $\xi$ attains values in $[0,1]$, is $0$ if and only if $X$ and $Y$ are independent and is $1$ if and only if $Y$ perfectly depends on $X$, that is, $Y$ is almost surely a function of $X$.
In contrast to concordance measures, $\xi$ quantifies directed dependence, i.e.~the impact of the random variable $X$ on the random variable $Y$. The two concepts are therefore fundamentally different.
We show that, for any choice of $S_\delta \in \mathcal{C}^{\rm LSL}$, 
$\xi(S_\delta)$ admits an explicit closed-form expression in $\delta$ and is always smaller than the corresponding values of Kendall's tau, Spearman's rho and Spearman's footrule, thereby complementing the discussion in \cite{ansari2025rockel}.
Finally, we prove the exact region determined by Chatterjee's rank correlation and Kendall's tau.

\section{Pairwise comparison of concordance measures}

Since the concordance measures in \eqref{LSL:MoC} depend solely on the copula diagonal, we begin by briefly reviewing the fundamental properties of $\delta \in \mathcal{D}^{\rm LSL}$. 
Building on this, we then prove the conjectured $\tau$-$\rho$ region for lower semilinear copulas, as originally proposed in \cite{maislinger2025}, and determine the corresponding $\rho$-$\phi$ and $\tau$-$\phi$ regions.

Copula diagonals are Lipschitz-continuous, and therefore absolutely continuous and differentiable $\lambda$-almost everywhere. 
As shown in \cite[Corollary 5]{durante2008}, the shape constraints in \eqref{Def:LSLDiagonal} are equivalent to the following chain of inequalities 
\begin{align} \label{Def:LSL2}
  \delta(t)\leq t \, \delta'(t) \leq 2 \, \delta(t)
\end{align}
which holds for $\lambda$-almost every $t \in [0,1]$.

Throughout this work, we make frequent use of the copula diagonals $u_a$ and $l_a$, introduced in \cite[Example 3.3]{maislinger2025}, which will play a central role in our analysis.

\begin{example} \label{ex:uala}
For \(a \in [0,1]\), define the functions \(u_a, l_a: [0,1] \to [0,1]\) by
\begin{align*} 
  u_a(t) 
  & := \begin{cases} \frac{t^2}{a} &\text{if } t\leq a \\ t &\text{if } t>a \end{cases} 
  & l_a(t) 
  & := \begin{cases} at &\text{if } t\leq a \\ t^2 &\text{if } t>a. \end{cases} 
\end{align*}
Then, \(u_a, l_a \in \mathcal{D}^{\rm LSL}\) and their corresponding values for \(\tau\), \(\rho\) and \(\phi\) are as follows:

\begin{center}
\begin{tabular}{l||ccc}
\toprule
 & $\tau(S_\delta)$ & $\rho(S_\delta)$ & $\phi(S_\delta)$ \\
\midrule
$\delta = u_a$ & $1-a^2$ & $1-a^3$ & $1-a^2$ \\
$\delta = l_a$ & $a^4$ & $a^4$ & $a^3$ \\
\bottomrule
\end{tabular}
\end{center}

\begin{figure}[h]
  \centering
  \begin{minipage}[b]{0.45\linewidth}
    \includegraphics[width=0.9\linewidth]{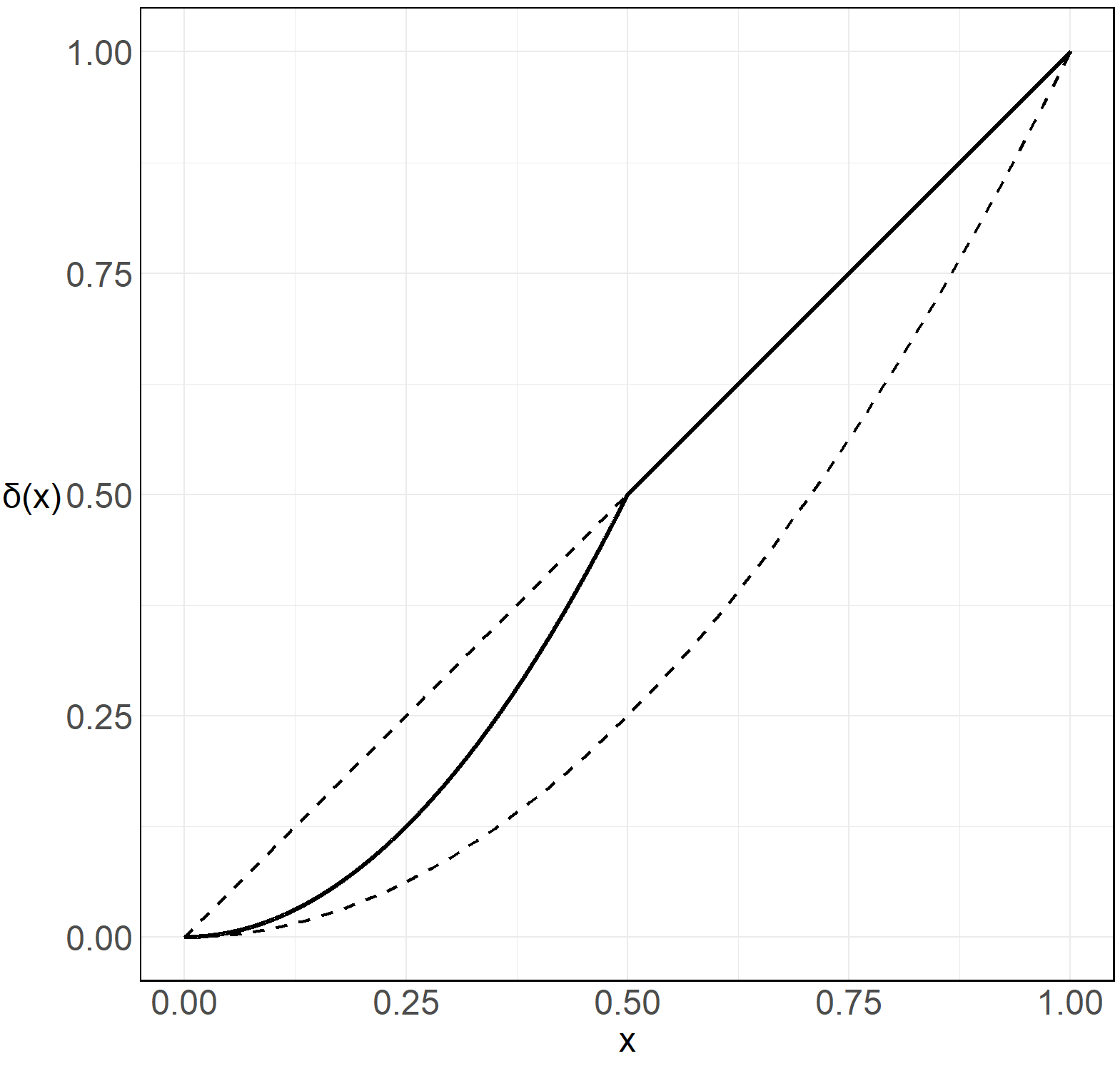}
  \end{minipage}
\hfill
  \begin{minipage}[b]{0.45\linewidth}
    \includegraphics[width=0.9\linewidth]{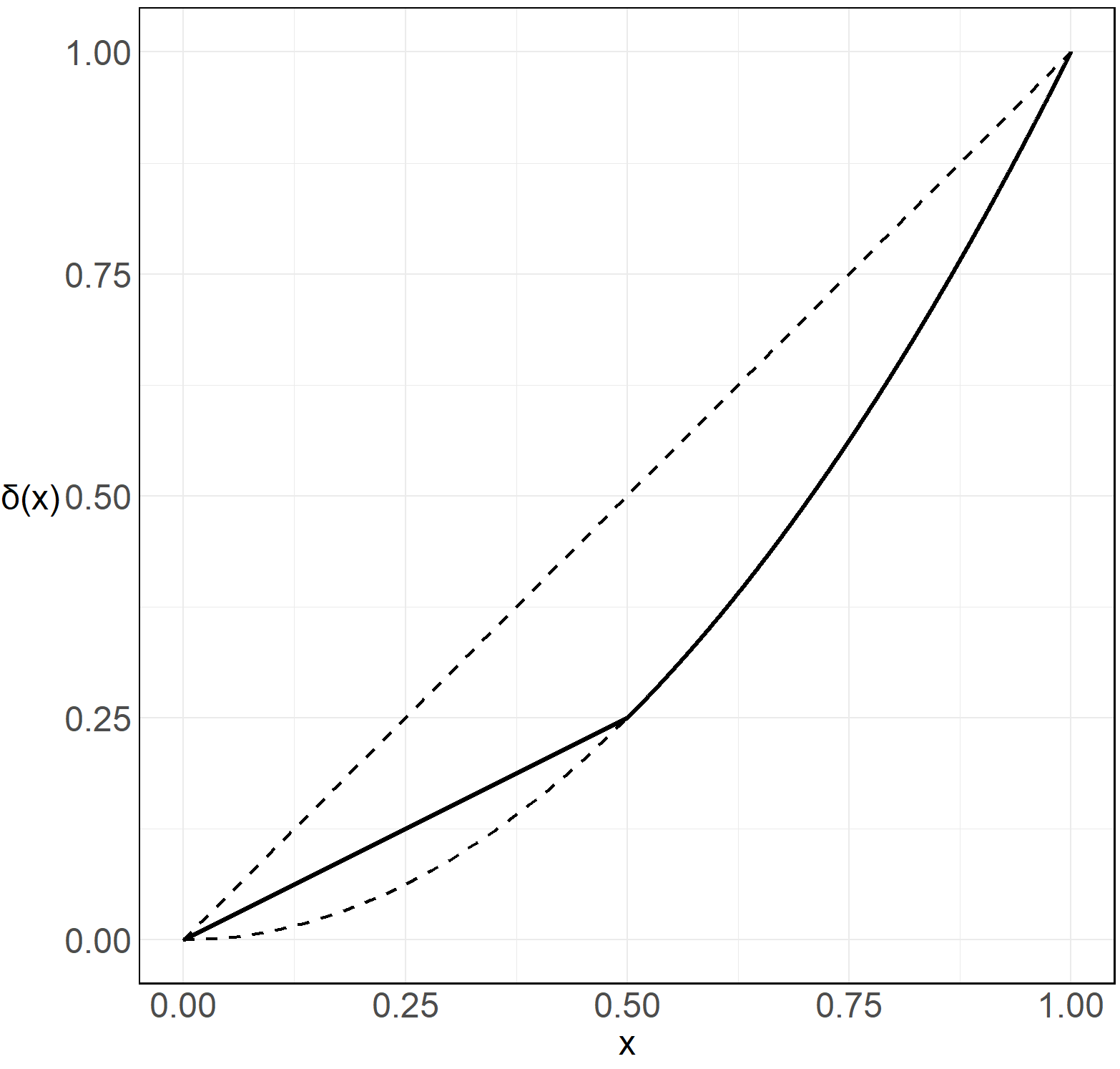}
  \end{minipage}
\caption{The copula diagonals \(\delta = u_a\) (left panel) and \(\delta=l_a\) (right panel) for \(a=\frac{1}{2}\) as considered in Example \ref{ex:uala}.}
\end{figure}
\end{example}

The next lemma relates the functions \(u_a\) and \(l_a\) to the inequalities in \eqref{Def:LSL2} and we will make use of this result in the following subsections.

\begin{lemma} \label{lem:ua}\label{lem:la}
Consider a copula diagonal \(\delta \in \mathcal{D}^{\rm LSL}\).
Then 
\begin{enumerate}
\item\label{lem:ua1} 
\(\delta=u_a\) for some \(a \in [0,1]\) if and only if, 
for $\lambda$-almost every \(t \in [0,1]\), either \(\delta(t)=t\) or \(t \,\delta'(t) =2\, \delta(t)\).
\item\label{lem:la1}
\(\delta=l_a\) for some \(a \in [0,1]\) if and only if,
for $\lambda$-almost every \(t \in [0,1]\), either \(\delta(t)=t^2\) or \(t \, \delta'(t) = \delta(t)\).
\end{enumerate}
\end{lemma}
\begin{proof}
We first prove \ref{lem:ua1}. Verifying the properties for function $u_a$ is evident.
For the other direction, set
\(a := \inf \left \{t \in [0,1], \delta(t)=t\right\} \in [0,1]\). 
Then, \(t \,\delta'(t) = 2\, \delta(t)\) for $\lambda$-almost every \(t \in [0,a)\), hence $t \mapsto \delta(t)/t^2$ is constant on $[0,a]$ implying $\delta(t) = t^2/a$ on $[0,a]$.
Since $\frac{\delta(a)}{a} = 1 = \frac{\delta(1)}{1}$ and $t \mapsto \frac{\delta(t)}{t}$ is non-decreasing due to \eqref{Def:LSLDiagonal}, we finally obtain $\delta(t) = t$ on $(a,1]$.
A similar line of reasoning yields assertion \ref{lem:la1}.
\end{proof}

We now introduce several additional examples that will be revisited in the subsequent sections. In particular, we construct a diagonal that satisfies, in a piecewise manner, either the lower or the upper bound in \eqref{Def:LSL2}.

\begin{example} \label{ex:otherdiag}
Define the function $\delta:[0,1] \to [0,1]$ by 
\[
  \delta(t)
  :=\begin{cases}
  \frac{8t^2}{3} &\text{if } t \in \left[0, \frac{1}{4}\right] \\
  \frac{2t}{3}   &\text{if } t \in \left(\frac{1}{4}, \frac{1}{2}\right] \\
  \frac{4t^2}{3} &\text{if } t \in \left(\frac{1}{2}, \frac{3}{4}\right] \\
  t              &\text{if } t \in \left(\frac{3}{4}, 1\right].
\end{cases}
\]
Then $\delta \in \mathcal{D}^{\rm LSL}$ with $t \, \delta^\prime(t) = \delta(t)$ on $(1/4,1/2) \cup (3/4,1)$ and $t \, \delta^\prime(t) = 2 \, \delta(t)$ on $(0,1/4) \cup (1/2,3/4)$.
\begin{figure}[h]
    \centering
    \includegraphics[width=0.40\textwidth]{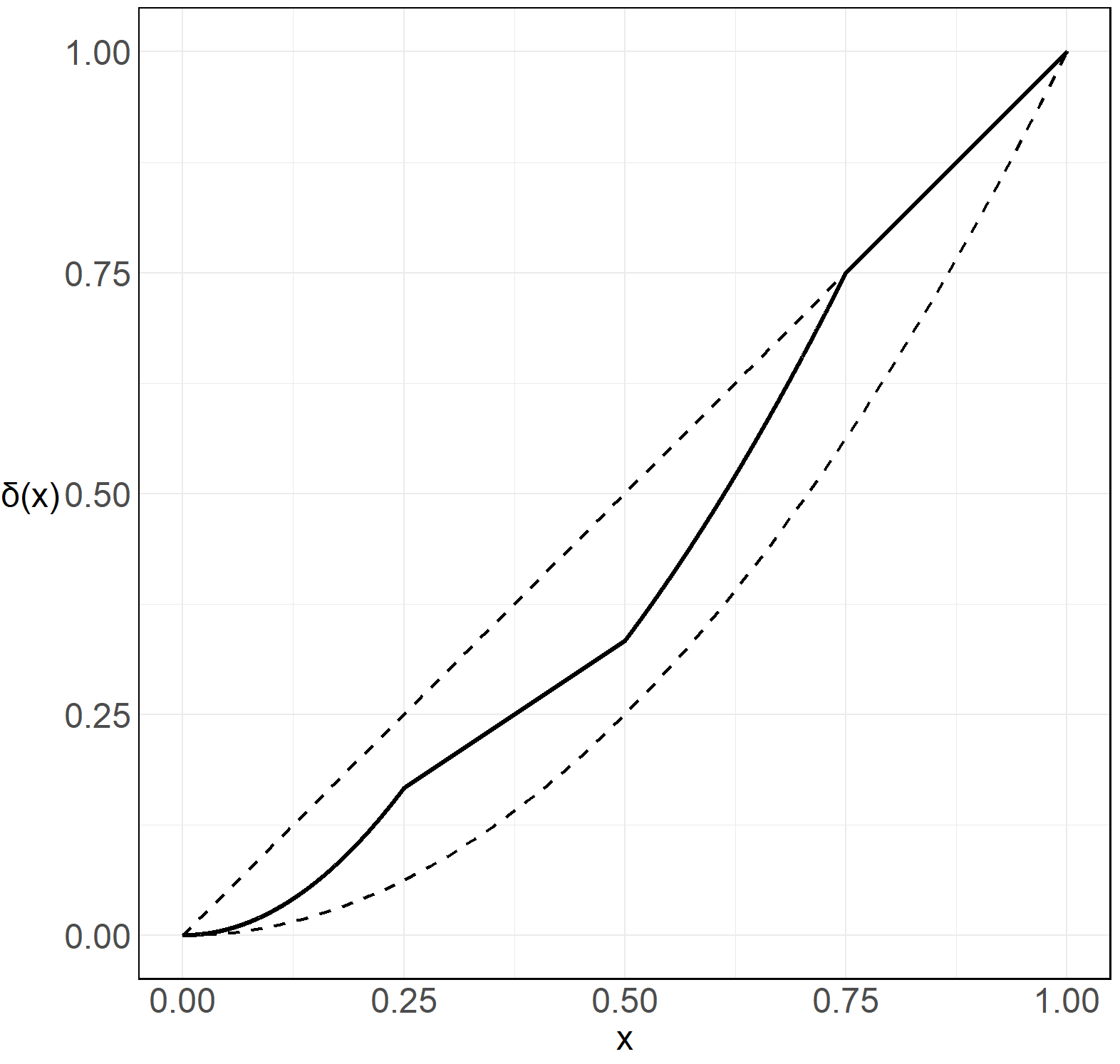}
\caption{The copula diagonal considered in Example \ref{ex:otherdiag}}
\end{figure}
\end{example}

\begin{example}[Fr{\'e}chet diagonals; \cite{durante2008}] \label{ex:frechet}
For $\alpha \in [0,1]$, consider the copula diagonal $\delta_{F;\alpha}$ given by 
$\delta_{F;\alpha}(t) := \alpha \, \delta_M + (1-\alpha) \, \delta_\Pi$,
where $\delta_M, \delta_\Pi: [0,1] \to [0,1]$ are the copula diagonals defined by $\delta_\Pi(t) = t^2$ and $\delta_M(t) = t$. 
Since $\delta_\Pi, \delta_M \in \mathcal{D}^{\rm LSL}$ and $\mathcal{D}^{\rm LSL}$ is convex, $\delta_{F;\alpha} \in \mathcal{D}^{\rm LSL}$ as well, with 
$t \, \delta^\prime_{F;\alpha}(t) = \delta_{F;\alpha}(t) + (1-\alpha) \, \delta_\Pi (t) = 2\, \delta_{F;\alpha}(t) - \alpha \, \delta_M (t)$ for all $t \in (0,1)$.
The associated lower semilinear copulas $S_{\delta_\alpha}$ are called Fr{\'e}chet copulas which are just convex linear combinations of the copulas $\Pi$ and $M$; see \cite[Example 7]{durante2008}.
The corresponding values for \(\tau\), \(\rho\) and \(\phi\) are as follows:
\begin{center}
\begin{tabular}{l||ccc}
\toprule
 & $\tau(S_{\delta_\alpha})$ & $\rho(S_{\delta_\alpha})$ & $\phi(S_{\delta_\alpha})$ \\
\midrule
$\delta_\alpha$ & $\alpha \, \frac{\alpha+2}{3}$ & $\alpha$ & $\alpha$ \\
\bottomrule
\end{tabular}
\end{center}
\end{example}

\begin{example}[Power diagonals] \label{ex:power}
For $p \in [1, 2]$, define the function $\delta_p: [0,1] \to [0,1]$ by 
$\delta_p(t) := t^p$.
Then $\delta_p \in \mathcal{D}^{\rm LSL}$ with $t \, \delta^\prime_p(t) = p \, \delta_p(t)$ for all $t \in (0,1)$, and the corresponding values for $\tau, \rho \text{ and } \phi$ are as follows:
    \begin{center}
    \begin{tabular}{l||ccc}
    \toprule
    & $\tau(S_{\delta_p})$ & $\rho(S_{\delta_p})$
    & $\phi(S_{\delta_p})$ \\
    \midrule
    $\delta_p $ & $\frac{2-p}{p}$ & $ 3 \frac{2-p}{p+2}$ & $ 2 \frac{2-p}{p+1}$ \\
    \bottomrule
    \end{tabular}
    \end{center}
    \end{example}

Finally, we discuss a well-studied class of copulas whose copula diagonals satisfy the shape constraints in \eqref{Def:LSL2}, but generally fail to be lower semilinear.

\begin{example}[Marshall-Olkin diagonals] \label{ex:MO}
For $\alpha, \beta \in [0,1]$, consider the copula diagonal $\delta_{MO;\alpha,\beta}$  given by 
\begin{align*}
  \delta_{MO;\alpha,\beta}(t)
  := t^{2 - \min\{\alpha,\beta\}}.
\end{align*}
Then, $\delta_{MO;\alpha,\beta} = \delta_{p}$ with $p=2 - \min\{\alpha,\beta\}$ (Ex.~\ref{ex:power}), hence $ \delta_{MO;\alpha,\beta} \in \mathcal{D}^{\rm LSL}$.
However, Marshall-Olkin copulas $M_{\alpha,\beta}$ \cite{durante2016} defined by 
\begin{align*}
  M_{\alpha,\beta}(u,v)
  := \min\{u^{1-\alpha}\,v,u\,v^{1-\beta}\}
\end{align*}
generally fail to be lower semilinear \cite[p.~10]{maislinger2025}, i.e.~$M_{\alpha,\beta}$ and $S_{\delta_{MO;\alpha,\beta}}$ generally differ.
If, instead, $\alpha = \beta$, then 
\begin{align*}
  M_{\alpha,\alpha}(u,v)
  & 
    = 
    \begin{cases}
         v \; \frac{\delta_{MO;\alpha,\alpha}(u)}{u} & \textrm{ if } v \leq u \\
         u \; \frac{\delta_{MO;\alpha,\alpha}(v)}{v} & \textrm{ otherwise }
    \end{cases}
    = S_{\delta_{MO;\alpha,\alpha}}(u,v)
    = S_{\delta_{2 - \alpha}}(u,v)
\end{align*}
for all $(u,v) \in [0,1]^2$. In other words, $M_{\alpha,\alpha}$ is lower semilinear. 
For the sake of completeness, we here list the values of \(\tau\), \(\rho\) and \(\phi\) for general Marshall-Olkin copulas $M_{\alpha,\beta}$ (see \cite{durante2016}):
\begin{center}
\begin{tabular}{ccc}
\toprule
$\tau(M_{\alpha,\beta})$ & $\rho(M_{\alpha,\beta})$ & $\phi(M_{\alpha,\beta})$ \\
\midrule
$\frac{\alpha\,\beta}{\alpha - \alpha\,\beta + \beta}$ & $\frac{3 \, \alpha\,\beta}{2 \, \alpha - \alpha\,\beta + 2\, \beta}$ & $\frac{2 \, \min\{\alpha,\beta\}}{3 - \min\{\alpha,\beta\}}$ \\
\bottomrule
\end{tabular}
\end{center}
\end{example}

\begin{remark}
Recalling \cite{maislinger2025}, for each $\delta \in \mathcal{D}^{\rm LSL}$ there exists a Borel set $\Lambda$ with $\lambda(\Lambda)=1$ and a measurable function $\overline{\Delta}: [0,1] \to [0,2]$ such that $\overline{\Delta}(t) = \delta'(t)$ for all $t \in \Lambda$. 
Now, define the function $\Delta: [0,1] \mapsto [0,2]$ by 
\begin{align*}
  \Delta(t) := \overline{\Delta}(t) \, \mathds{1}_{\Lambda} (t) + \frac{\delta(t)}{t} \, \mathds{1}_{\Lambda^c} (t)\,,
\end{align*}
with the convention that $\tfrac{0}{0} := 0$.
Then $\Delta$ coincides with $\delta'$ on $\Lambda$, 
hence it is a measurable version of $\delta'$ that is defined on the whole interval $[0,1]$.
Since in what follows we mainly work with Lebesgue integrals, we shall not distinguish between different versions of the diagonal derivative, as they coincide $\lambda$-almost everywhere.
\end{remark}

\subsection{The exact \(\tau\)-\(\rho\) region for lower semilinear copulas}

In \cite{maislinger2025}, the authors studied the \( \tau \)-\( \rho \) region $\Omega_{\tau,\rho}^{\rm LSL}$ stated in \eqref{Conjecture} for the class of lower semilinear copulas. 
They proved the inequality 
$$\tau(S_\delta)\leq \rho(S_\delta)$$
and further conjectured that 
$$(1-\tau(S_\delta))^3 \leq (1-\rho(S_\delta))^2.$$
Here, we complete the discussion by providing a proof for the conjectured second inequality and a new proof for the first inequality that avoids the use of Markov kernels.

\begin{theorem} \label{Thm:taurho}
The inequalities
\begin{align*}
  \tau(S_\delta) \leq \rho(S_\delta) \leq 1- (1- \tau(S_\delta))^{3/2}    
\end{align*}
hold for every \(S_\delta \in \mathcal{C}^{\rm LSL}\).
Moreover, both inequalities are sharp.
\end{theorem}
\begin{proof}
We first provide an alternative proof for $\tau(S_\delta) \leq \rho(S_\delta)$.
Straightforward calculation yields
\begin{align*}
  \tau(S_\delta)-\rho(S_\delta)
  & = \left( 4 \, \int_{[0,1]} \frac{\delta^2(t)}{t} \de \lambda(t) - 1 \right) - \left( 12 \, \int_{[0,1]} t \, \delta(t) \de \lambda(t) - 3 \right) 
  \\
  & = 4 \, \int_{[0,1]} \frac{\delta^2(t)}{t} - \frac{t^2 \, \delta(t)}{t} - 2\,t\,\delta(t) + 2t^3 \de \lambda(x) 
  \\
  & = 4 \, \int_{[0,1]} \underbrace{(\delta(t)-t^2)}_{\eqref{Def:LSLDiagonal} \; \geq 0} \left(\underbrace{\frac{\delta(t)}{t}}_{\eqref{Def:LSL2} \; \leq \delta'(t)}-2t\right) \de \lambda(t) 
  \\
  & \leq 4 \, \int_{[0,1]} \underbrace{(\delta(t)-t^2)}_{=: f(t)} \, \underbrace{(\delta'(t)-2t)}_{= f^\prime(t)} \de \lambda(t) = 0\,.
\end{align*}
Since $f$ is Lipschitz continuous on $[0,1]$, hence absolutely continuous, and $f^\prime$ is integrable, 
the last identity follows from integration by parts (see \cite[p.64]{pap2002}) incorporating the identities $\delta(1) = 1$ and $\delta(0) = 0$.

We now prove the upper inequality. 
We begin by reformulating the original problem and then apply Hölder inequality.
Therefore, define the function $g: [0,1] \to [0,1]$ given by 
\begin{align*}
  g(t) 
	& := \begin{cases}
			 \frac{\delta(t)}{t}                & t \in (0,1] \\
			 \lim_{t \to 0} \frac{\delta(t)}{t} & t=0
			 \end{cases}
\end{align*}
and the two terms
\begin{align*}
  A 
	& := \frac{1}{2} - \int_{(0,1]} \frac{\delta^2(t)}{t} \de \lambda(x)
	   = \int_{(0,1]} t \, (1 - g^2(t)) \de \lambda(t)
	\\
	B 
	& := \frac{1}{3} - \int_{(0,1]} t \, \delta(t) \de \lambda(t)
	   = \int_{(0,1]} t^2 \, (1 - g(t)) \de \lambda(t)\,.
\end{align*}
Then $g$ is continuous and non-decreasing due to \eqref{Def:LSLDiagonal}, $t \leq g(t) \leq 1$ for all $t \in [0,1]$, and $g(1)=1$ as well as $g(0) \geq 0$.
Applying integration by parts \cite{billingsley-1995} to $A$ and $B$ and introducing the measure $\mu: \mathcal{B}((0,1]) \to [0,\infty]$ given by 
\begin{align*}
  \mu(S) 
  & := \int_{S} t^3 \, g'(t) \de \lambda(t)
\end{align*} 
further gives
\begin{align*}
  A 
  & = \int_{(0,1]} t^2 \, g(t) \, g'(t) \de \lambda(t)
	  = \int_{(0,1]} \frac{g(t)}{t} \de \mu(t)
  \\
  B  
  & = \frac{1}{3} \, \int_{(0,1]} t^3 \, g'(t) \de \lambda(t)
	  = \frac{1}{3} \, \int_{(0,1]} 1 \de \mu(t)
		= \frac{1}{3} \; \mu((0,1])\,.
\end{align*}
Since $(1 - \tau(S_\delta))^3 = (4A)^3$ and $(1 - \rho(S_\delta))^2 = (12B)^2$, it thus remains to prove that $(4A)^3 \leq (12B)^2$, but this follows from Hölder inequality with $p=3$ and $q=3/2$ as follows
\begin{align*}
  (4A)^3
	&   =  4^3 \, \left( \int_{(0,1]} \frac{g(t)}{t} \cdot 1 \de \mu(t) \right)^3
	   \leq 4^3 \, \int_{(0,1]} \left(\frac{g(t)}{t}\right)^3 \de \mu(t) \cdot \mu([0,1])^2
	\\
	&   =  4^3 \, \int_{(0,1]} g(t)^3 \, g'(t) \de \lambda(t) \cdot (3B)^2\,.
\end{align*}
Using integration by parts \cite{billingsley-1995} a second time gives 
\begin{align*}
  \int_{[0,1]} g(t)^3 \, g'(t) \de \lambda(t)
	&   =  \frac{g(1)^4 - g(0)^4}{4} 
	  \leq \frac{1}{4}\,.
\end{align*}
Thus, $(4A)^3 \leq 4^2 \, (3B)^2 = (12 \, B)^2$. This proves the desired inequality.

Finally, the lower inequality is sharp for the copulas \(S_{l_a}\) and the upper inequality is sharp for the copulas \(S_{u_a}\) due to Example \ref{ex:uala}.
\end{proof}

\begin{figure}[h]
  \centering
  \includegraphics[width=0.45\textwidth]{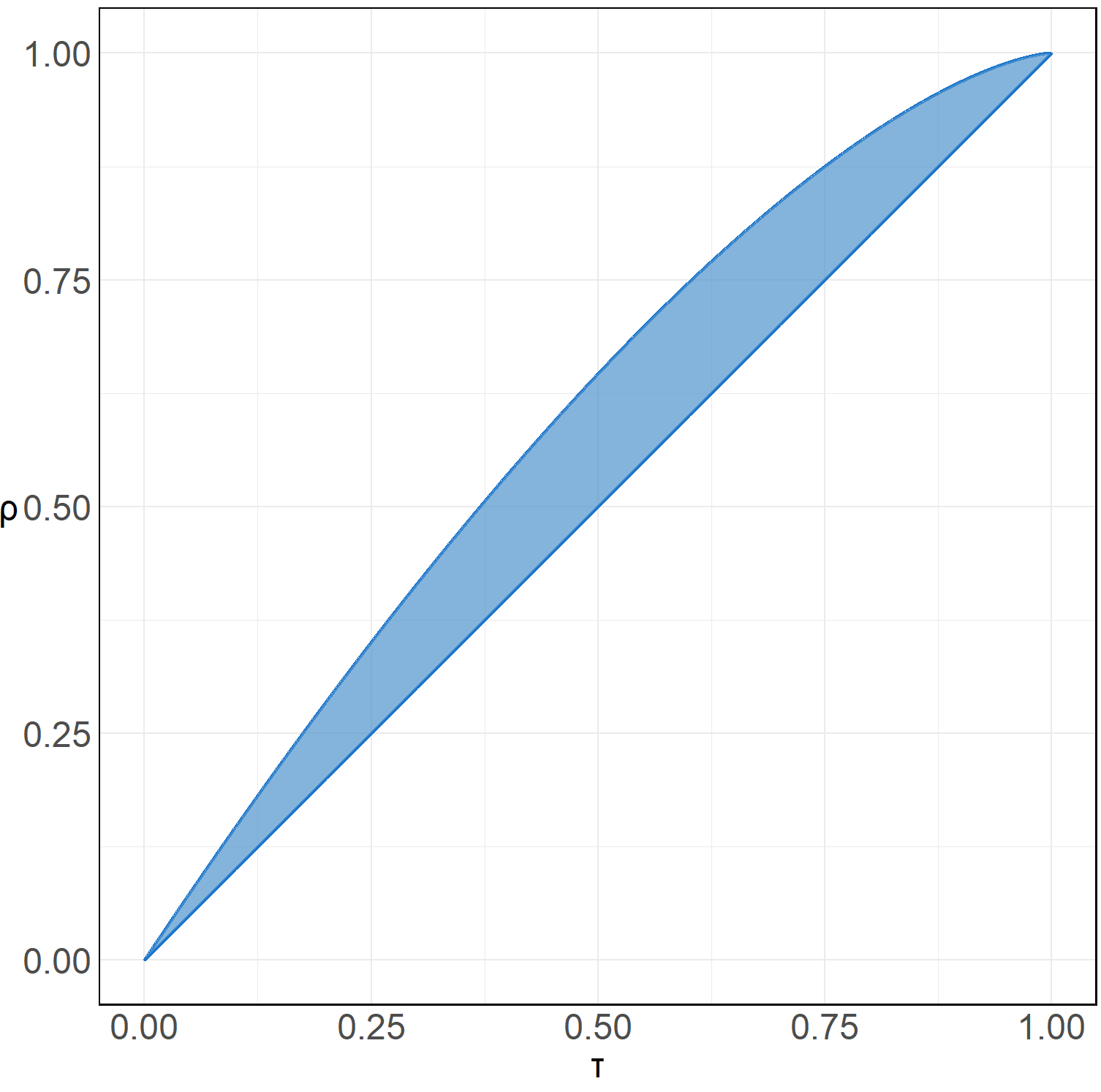}
\caption{The \(\tau\)-\(\rho\) region \( \Omega^{\rm LSL}_{\tau \text{,}\rho}\)}
\end{figure}

According to \cite{maislinger2025}, the \(\tau\)-\(\rho\) region is convex, compact and the lower inequality in Theorem \ref{Thm:taurho} is sharp exclusively for copulas of the form $S_{l_a}$ with copula diagonal $l_a$ discussed in Example \ref{ex:uala}.
In contrast, as stated in the proof of Theorem \ref{Thm:taurho}, the upper inequality is sharp for copulas \(S_{u_a}\).

\begin{remark}
In \cite{fredricks2007}, the authors have established the inequality $\tau(C) \leq \rho(C)$ for copulas $C$ that are left-tail decreasing (LTD), that is $u \, \partial_1 C(u,v) \leq C(u,v)$ for $\lambda$-almost all $u \in [0,1]$ and all $v \in [0,1]$, and right-tail increasing (RTI), that is $(1-u) \, \partial_1 C(u,v) \geq v - C(u,v)$ for $\lambda$-almost all $u \in [0,1]$ and all $v \in [0,1]$;
here, $\partial_1 C(u,v)$ denotes the partial derivate of $C$ with respect to the first argument (which, for each $v \in [0,1]$, exists for $\lambda$-almost all $u \in [0,1]$ \cite{durante2016}). 
Although lower semilinear copulas are LTD \cite[p.~6]{maislinger2025}, they typically fail to be RTI (take the lower semilinear copula with the copula diagonal in Example \ref{ex:otherdiag} for $v < u \in (1/4,1/2)$). 
Consequently, the inequality stated in Theorem \ref{Thm:taurho} is not encompassed by the corresponding result in \cite{fredricks2007}.
\end{remark}

\subsection{The exact \(\tau\)-\(\phi\) region for lower semilinear copulas}

We now determine the \(\tau\)-\(\phi\) region 
$\Omega^{\rm LSL}_{\tau \text{,}\phi}:=\{(\tau(S_\delta), \phi(S_\delta)): S_\delta \in \mathcal{C}^{\rm LSL}\}$ and prove that it coincides with the subset of the unit square given by
\[ \{(x,y) \in [0,1]^2: x\leq y \leq x^{\frac{3}{4}}\}.\]

\begin{theorem} \label{Thm:tauphi}
The inequalities
\begin{align*}
  \tau(S_\delta) \leq \phi(S_\delta) \leq (\tau(S_\delta))^{3/4}.
\end{align*}
hold for every \(S_\delta \in \mathcal{C}^{\rm LSL}\). Moreover, both inequalities are sharp.
\end{theorem}
\begin{proof}
We first prove the lower inequality $\tau(S_\delta) \leq \phi(S_\delta)$.
Straightforward calculation yields
\begin{align*}
  \phi(S_\delta)-\tau(S_\delta)
  & = \left( 6 \, \int_{[0,1]} \delta(t) \de \lambda(t) - 2 \right) - \left( 4 \, \int_{[0,1]} \frac{\delta^2(t)}{t} \de \lambda(t) - 1 \right) 
  \\
  & = \int_{[0,1]} 2 \, (t-\delta(t)) \, \frac{2\delta(t)}{t} - 2 \, (t-\delta(t)) \de \lambda(t) 
  \\
  & = \int_{[0,1]} 2\underbrace{(t-\delta(t))}_{\eqref{Def:LSLDiagonal} \; \geq 0}\left(\underbrace{\frac{2\delta(t)}{t}}_{\eqref{Def:LSL2} \; \geq \delta'(t)}-1\right) \de \lambda(t) 
  \\
  & \geq -2\int_{[0,1]} \underbrace{(t-\delta(t))}_{=:f(t)} \, \underbrace{(1-\delta'(t))}_{=f^\prime(t)} \de \lambda(t) = 0\,.
\end{align*}
Since $f$ is Lipschitz continuous on $[0,1]$, hence absolutely continuous, and $f^\prime$ is integrable, 
the last identity follows from integration by parts (see \cite[p.64]{pap2002}) incorporating the identities $\delta(1) = 1$ and $\delta(0) = 0$.

We now prove the upper inequality.
As in the proof of Theorem \ref{Thm:taurho}, we again begin by reformulating the original problem and then apply Hölder inequality.
Define the function $g: [0,1] \to [0,1]$ given by 
\begin{align*}
  g(x) 
	& := \begin{cases}
			 \frac{\delta(t)}{t^2}                & t \in (0,1] \\
			 \lim_{t \to 0} \frac{\delta(t)}{t^2} & t=0
			 \end{cases}
\end{align*}
and the two terms
\begin{align*}
  A 
	& := \int_{(0,1]} \delta(t) \de \lambda(t) - \frac{1}{3}
	   = \int_{(0,1]} t^2 \, (g(t) - 1) \de \lambda(t)
	\\
	B 
	& := \int_{(0,1]} \frac{\delta(t)^2}{t} \de \lambda(t) - \frac{1}{4}
	   = \int_{(0,1]} t^3 \, (g(t)^2 - 1) \de \lambda(t)\,.
\end{align*}
Then $g$ is continuous and non-increasing due to \eqref{Def:LSLDiagonal}, $1 \leq g(t) \leq 1/t$ for all $t \in [0,1]$, and $1/g(1)=1$ as well as $1/g(0) \geq 0$.
Applying integration by parts \cite{billingsley-1995} to $A$ and $B$ and introducing the measure $\mu: \mathcal{B}((0,1]) \to [0,\infty]$ given by 
\begin{align*}
  \mu(S) 
  & := \int_{S} \underbrace{\frac{-g'(t)}{g(t)^3}}_{\geq 0} \de \lambda(t)
\end{align*} 
further gives
\begin{align*}
  A 
  & = \frac{1}{3} \; \int_{(0,1]} t^3 \, (-g'(t)) \de \lambda(t)
	  = \frac{1}{3} \; \int_{(0,1]} t^3 \, g(t)^3 \de \mu(t)
  \\
  B  
  & = \frac{1}{2} \, \int_{(0,1]} t^4 \, g(t) \, (-g'(t)) \de \lambda(t)
	  = \frac{1}{2} \, \int_{(0,1]} t^4 \, g(t)^4 \de \mu(t)\,.
\end{align*}
Since $(\phi(S_\delta))^4 = (6A)^4$ and $(\tau(S_\delta))^3 = (4B)^3$, it thus remains to prove that $(6A)^4 \leq (4B)^3$, but this follows from Hölder inequality with $p=4/3$ and $q=4$ as follows
\begin{align*}
  (6A)^4
	&   =  6^4 \, \left(\frac{1}{3}\right)^4 \left( \int_{(0,1]} t^3 \, g(t)^3 \cdot 1 \de \mu(t) \right)^4
	   \leq 6^4 \, \left(\frac{1}{3}\right)^4 \left( \int_{(0,1]} t^4 \, g(t)^4 \de \mu(t) \right)^3 \cdot \mu((0,1])
	\\
	&   =  2^4 \, (2B)^3 \cdot \int_{(0,1]} \frac{-g'(t)}{g(t)^3} \de \lambda(t)
        =  2^4 \, (2B)^3 \cdot \frac{1}{2} \; \left(\frac{1}{g(1)^2} - \frac{1}{g(0)^2} \right)
      \leq 2^4 \, (2B)^3 \cdot \frac{1}{2}\,.
\end{align*}
Thus, $(6A)^4 \leq 2^4 \, (2B)^3 \cdot \frac{1}{2} = (4B)^3$. This proves the desired inequality.

Finally, the lower inequality is sharp for the copulas \(S_{u_a}\) and the upper inequality is sharp for the copulas \(S_{l_a}\) due to Example \ref{ex:uala}.
\end{proof}

The next result characterizes the set of copulas for which the lower inequality in Theorem \ref{Thm:tauphi} is sharp.

\begin{corollary} \label{tauphi:sharp}
Within the class of lower semilinear copulas, 
the identity \(\tau(S_\delta)=\phi(S_\delta)\) is attained exclusively for the copulas of the form \(S_{u_a}\) with copula diagonal $u_a$.
\end{corollary}
\begin{proof}
According to the proof of Theorem \ref{Thm:tauphi}, 
\(\tau(S_\delta) = \phi(S_\delta)\) if and only if 
\begin{align*}
  0
  & = 2 \, \int_{[0,1]} (t-\delta(t)) \, \left(\frac{2\delta(t)}{t}-1\right) \de \lambda(t) + 2 \, \int_{[0,1]} (t-\delta(t)) \, (1-\delta'(t)) \de \lambda(t)
  \\
  & = 2 \, \int_{[0,1]} (t-\delta(t)) \, \left(\frac{2\delta(t)}{t} - \delta'(t)\right) \de \lambda(t)\,.
\end{align*}
Since both factors are non-negative, the term is $0$ if and only if, for $\lambda$-almost every $t \in [0,1]$,
either 
$t=\delta(t)$ or $2\,\delta(t) = t\,\delta'(t)$, but this is equivalent to \(\delta=u_a\) for some \(a \in [0,1]\) according to Lemma \ref{lem:ua}.
\end{proof}

In contrast, as stated in the proof of Theorem \ref{Thm:tauphi}, the upper inequality is sharp for copulas \(S_{l_a}\).

The proof of the following corollary is conducted analogously to \cite[Theorem 5.7]{maislinger2025} and is therefore omitted.

\begin{corollary}
The \(\tau\)-\(\phi\) region $\Omega^{\rm LSL}_{\tau \text{,}\phi}$ is convex and compact.
\end{corollary}

\subsection{The exact \(\phi\)-\(\rho\) region for lower semilinear copulas}

We conclude this section by determining the \(\phi\)-\(\rho\) region 
$\Omega^{\rm LSL}_{\phi \text{,}\rho}:=\{(\phi(S_\delta), \rho(S_\delta)): S_\delta \in \mathcal{C}^{\rm LSL}\}$ and prove that it coincides with the set 
\[\{(x,y) \in [0,1]^2, x^\frac{4}{3}\leq y \leq 1-(1-x)^{\frac{3}{2}}\}.\]
The inequalities immediately follow from Theorems \ref{Thm:taurho} and \ref{Thm:tauphi}:
On the one part
\[ 
  \phi(S_\delta)^{4/3} \leq \tau(S_\delta) \leq \rho(S_\delta)\,.
\]
This inequality is sharp for the copulas \(S_{l_a}\) due to Example \ref{ex:uala}.
On the other hand, 
\[
  \rho(S_\delta) \leq 1-(1-\tau(S_\delta))^\frac{3}{2} \leq 1-(1-\phi(S_\delta))^\frac{3}{2}.
\]
This inequality is sharp for the copulas \(S_{u_a}\) due to Example \ref{ex:uala}.
We summarize our recent findings.

\begin{theorem} \label{Thm:phirho}
The inequalities
\begin{align*}
  \phi(S_\delta)^{4/3} \leq \rho(S_\delta) \leq 1-(1-\phi(S_\delta))^\frac{3}{2}
\end{align*}
hold for every \(S_\delta \in \mathcal{C}^{\rm LSL}\). Moreover, both inequalities are sharp.
\end{theorem}

Convexity of the region $\Omega^{\rm LSL}_{\phi \text{,}\rho}$ immediately follows from the fact that both \(\phi\) and \(\rho\) (unlike \(\tau\)) are compatible with convex combinations.

\begin{corollary}
The \(\phi\)-\(\rho\) region $\Omega^{\rm LSL}_{\phi \text{,}\rho}$ is convex and compact.
\end{corollary}

\begin{figure}[h]
  \centering
  \includegraphics[width=0.45\textwidth]{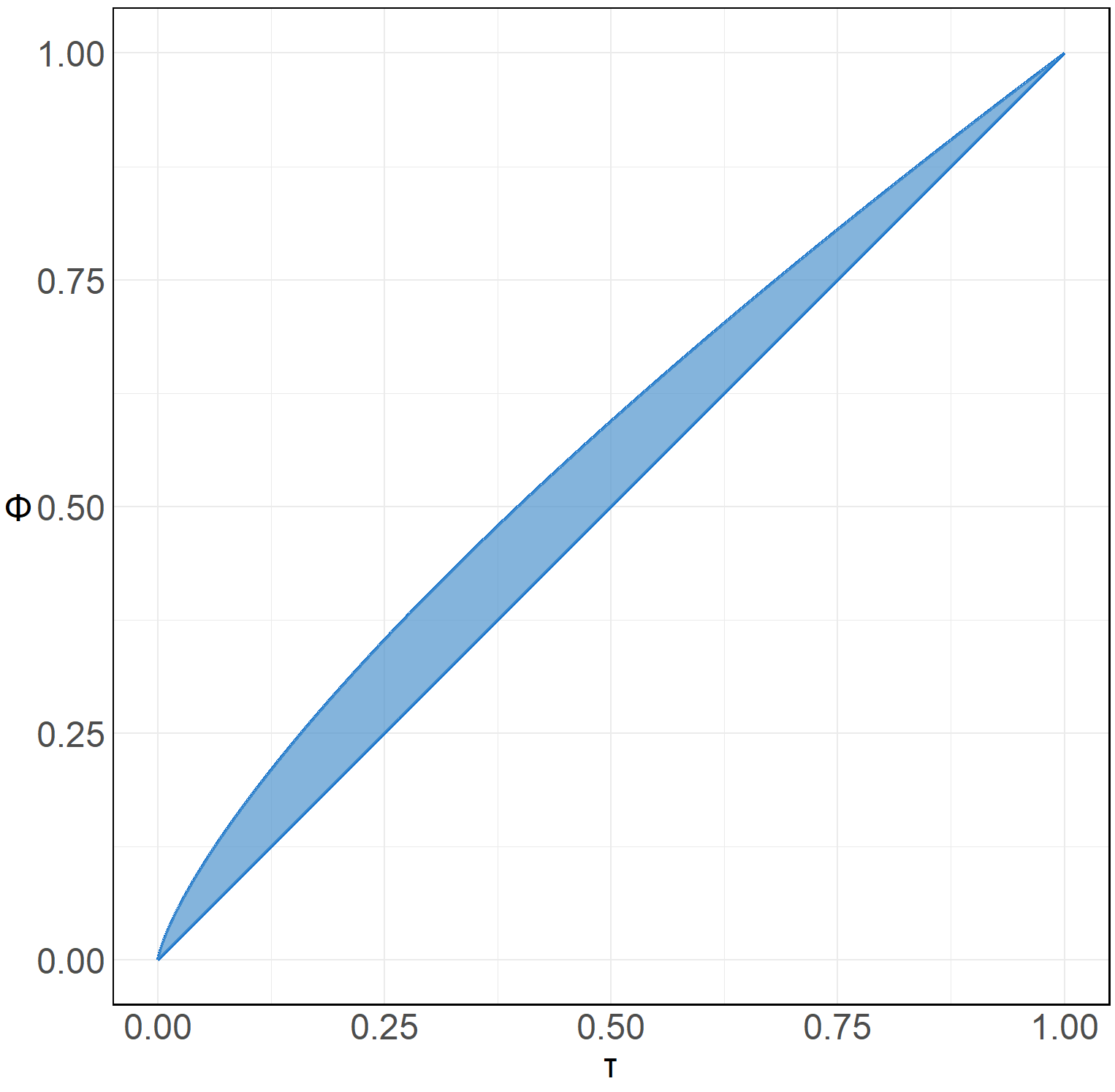}
  \includegraphics[width=0.445\textwidth]{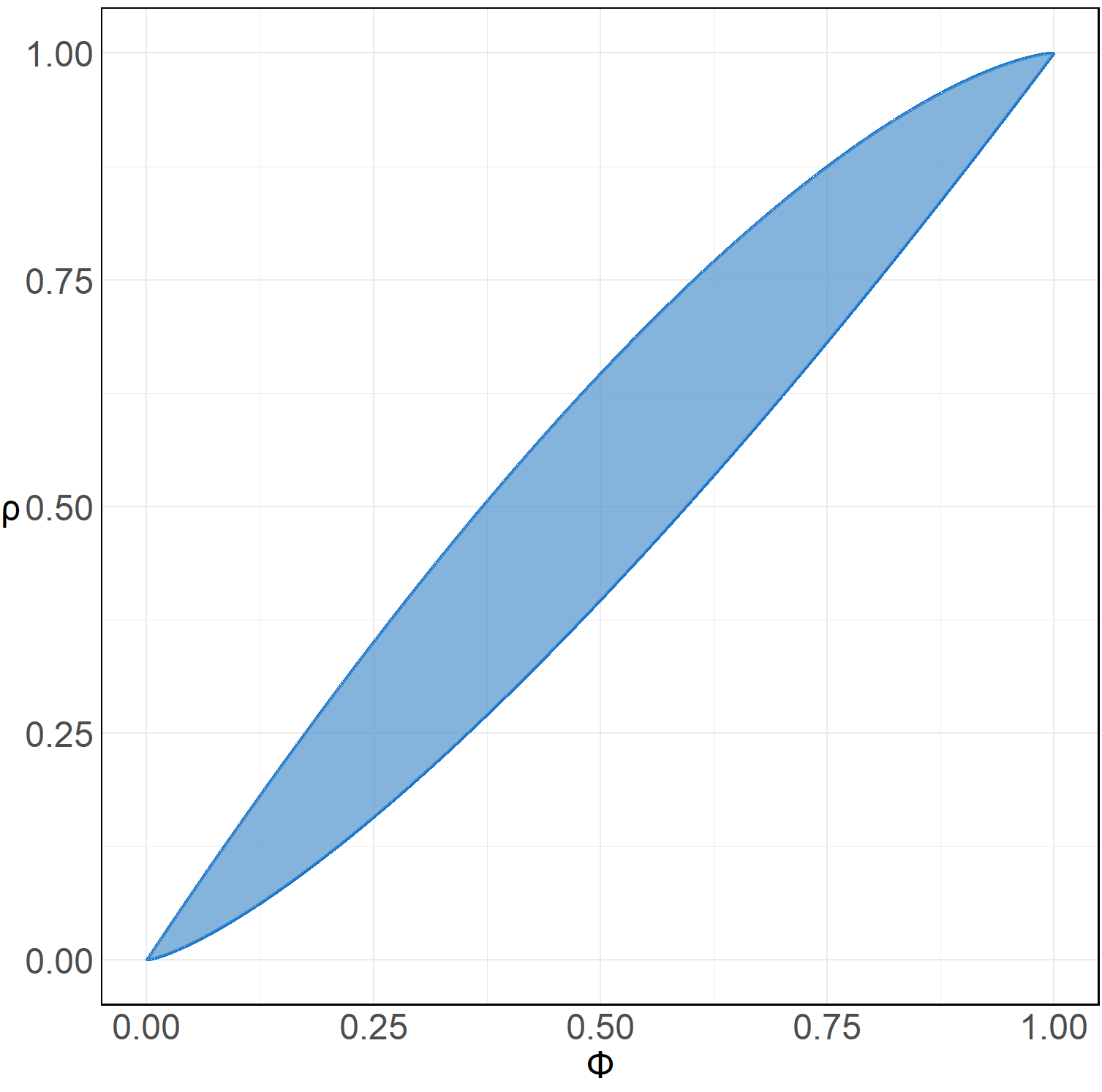}
  \caption{The \(\tau\)-\(\phi\) region \( \Omega^{\rm LSL}_{\tau \text{,}\phi}\) (left panel) and the \(\phi\)-\(\rho\) region \( \Omega^{\rm LSL}_{\Phi \text{,}\rho}\) (right panel).}
\end{figure}

\begin{remark}
The established \(\tau\)-\(\rho\), \(\tau\)-\(\phi\), and \(\phi\)-\(\rho\) regions are valid for lower semilinear copulas $S_\delta \in \mathcal{C}^{\rm LSL}$. 
However, for copulas whose copula diagonal $\delta$ belongs to the class $\mathcal{D}^{\rm LSL}$ but which are not themselves lower semilinear, the corresponding inequalities may not hold.
As a counterexample, consider the Marshall-Olkin copulas discussed in Example \ref{ex:MO}: 
For $\alpha = 1/10$ and $\beta = 1$, we have $\phi(M_{\alpha,\beta}) = 2/29 < 1/10 = \tau(M_{\alpha,\beta})$. This contradicts Theorem \ref{Thm:tauphi}.
\end{remark}

\section{Pairwise comparison of Chatterjee's rank correlation and measures of concordance}

It has been shown in \cite{siburg2013, fuchs2023JMVA} that Chatterjee's rank correlation for a copula $C$ in \eqref{Def:Chatterjee} equals Spearman's footrule of the so-called Markov product of $C$ and its transpose $C^T$, namely
\begin{align*}
  (C^T \ast C)(u,v) 
  & := \int_{[0,1]} \partial_1 C(t,u) \cdot \partial_1 C(t,v) \de \lambda(t)\,;
\end{align*}
we refer to \cite{durante2016} for more information on the partial derivatives of $C$ and the Markov product $C^T \ast C$.
Interestingly, the Markov product of a lower semilinear copula is again lower semilinear; \cite[Theorem 4.1]{maislinger2025}.
Since lower semilinear copulas are symmetric, we therefore have
\begin{align} \label{Rep:Chatterjee}
  \xi(S_\delta) 
  & = \phi(S_\delta \ast S_\delta)
    = 6 \, \int_{[0,1]} (S_\delta \ast S_\delta)(t,t) \de \lambda(t) - 2\,,
\end{align}
and we denote by $\delta^\ast (t) := (S_\delta \ast S_\delta)(t,t)$ the copula diagonal of the Markov product $S_\delta \ast S_\delta$. Then $\delta^\ast \in \mathcal{D}^{\rm LSL}$; \cite[Theorem 4]{durante2008}.
This implies, in particular, that Chatterjee's rank correlation $\xi(S_\delta)$ of $S_\delta$ depends solely on the diagonal entries of the associated Markov product; we refer the reader to \cite{fuchs2025DeMo} for more information on the Markov product and its role in quantifying directed dependence.

\subsection{The exact \(\tau\)-\(\xi\) region for lower semilinear copulas}

The next theorem provides an explicit representation of Chatterjee's rank correlation in terms of Kendall's tau and an additional term incorporating the shape constraints in  \eqref{Def:LSL2}.

\begin{theorem}\label{Rep:Chatterjee.2}
For \(S_\delta \in \mathcal{C}^{\rm LSL}\), Chatterjee's rank correlation fulfills
\begin{align*}
  \xi(S_\delta) 
  & = \tau(S_\delta)
     - 2 \, \int_{(0,1]} \frac{(t \, \delta^\prime(t) - \delta(t)) \, (2 \, \delta(t) - t \, \delta^\prime(t))}{t} \de \lambda(t)\,.
\end{align*}
\end{theorem}
\begin{proof}
Combining the above representation in \eqref{Rep:Chatterjee} and the diagonal expression of the Markov product for lower semilinear copulas presented in \cite[Theorem 4.1]{maislinger2025} yields
\begin{align*}
  \xi(S_\delta)
  & = 6 \, \int_{(0,1]} \frac{\delta(t)^2}{t} + t^2 \, 
        \int_{[t,1]} \left(\left(\frac{\delta(s)}{s}\right)^\prime \right)^2 \de \lambda(s) \de \lambda(t) - 2 
  \\
  & = 6 \, \int_{(0,1]} \frac{\delta(t)^2}{t} \de \lambda(t)
     + 6 \, \int_{(0,1]} \left( \int_{[0,s]} t^2 \de \lambda(t) \right) \left(\left(\frac{\delta(s)}{s}\right)^\prime \right)^2 \de \lambda(s) - 2 
  \\
  & = 4 \, \int_{(0,1]} \frac{\delta(t)^2}{t} \de \lambda(t) - 1
     + 6 \, \int_{(0,1]} \frac{s^3}{3} \left( \frac{s \, \delta^\prime(s) - \delta(s)}{s^2} \right)^2 \de \lambda(s) + 2 \, \int_{(0,1]} \frac{\delta(t)^2}{t} \de \lambda(t) - 1
  \\
  & = \tau(S_\delta)
     + 2 \, \int_{(0,1]} \frac{(s \, \delta^\prime(s) - \delta(s))^2}{s} \de \lambda(s) + 2 \, \int_{(0,1]} \frac{\delta(t)^2}{t} \de \lambda(t) - 2 \, \int_{(0,1]} \delta^\prime(t) \, \delta(t) \de \lambda(t)
  \\
  & = \tau(S_\delta)
     + 2 \, \int_{(0,1]} \frac{s^2 \, (\delta^\prime(s))^2 - 3 \, s \, \delta^\prime(s) \, \delta(s) + 2 \, \delta(s)^2}{s} \de \lambda(s) 
  \\
  & = \tau(S_\delta)
     - 2 \, \int_{(0,1]} \frac{(s \, \delta^\prime(s) - \delta(s)) \, (2 \, \delta(s)) - s \, \delta^\prime(s))}{s} \de \lambda(s)\,, 
\end{align*}
where the second identity follows from Fubini's theorem and the fourth identity is due to integration by parts \cite{billingsley-1995}.
\end{proof}

As a direct consequence of the shape constraints in \eqref{Def:LSL2}, 
the second term on the right-hand side of Theorem \ref{Rep:Chatterjee.2} is non-negative.
Consequently, Chatterjee's rank correlation never exceeds Kendall's tau.
We now determine the exact $\tau$-$\xi$ region $\Omega_{\tau,\xi}^{\rm LSL} := \{(\tau(S_\delta),\xi(S_\delta)): S_\delta \in \mathcal{C}^{\rm LSL}\}$.

\begin{theorem} \label{thm:markovtau}
The inequalities 
\[
  \frac{2 \, \tau(S_\delta)^2}{1+\tau(S_\delta)}
  \leq \xi(S_\delta)
  \leq \tau(S_\delta)
\]
hold for all \(S_\delta \in \mathcal{C}^{\rm LSL}\). Moreover, both inequalities are sharp.
\end{theorem}
\begin{proof}
The upper inequality is evident and sharp for diagonals $l_a$ and $u_a$ for any choice of $a \in [0,1]$ due to Lemma \ref{lem:ua}.
We now prove the lower inequality.
Therefore, define the two terms
\begin{align*}
  A 
  & := \int_{(0,1]} \frac{\delta(t)^2}{t} \de \lambda(t)
  & B
  & := \int_{(0,1]} \frac{(t \, \delta^\prime(t) - \delta(t)) \, (2 \, \delta(t) - t \, \delta^\prime(t))}{t} \de \lambda(t)\,.
\end{align*}
Then $\tau(S_\delta) = 4\, A - 1$ and $\xi(S_\delta) = \tau(S_\delta) - 2\,B$.
Further, define the function $g: (0,1] \to \mathbb{R}$ given by
\begin{align*}
  g(t) := \frac{t \, \delta^\prime(t)}{\delta(t)} \in [1,2] \quad \textrm{ due to \eqref{Def:LSL2}}\,,
\end{align*}
and the measure $\mu: \mathcal{B}((0,1]) \to [0,\infty]$ given by 
\begin{align*}
  \mu(S) 
  & := \frac{1}{A} \, \int_{S} \frac{\delta(t)^2}{t} \de \lambda(t)\,.
\end{align*}
Then, applying Jensen's inequality to the concave function $h: (0,\infty) \to \mathbb{R} $ such that $h(z) := (z-1) \, (2-z)$ yields
\begin{align*}
  B
  &   =  \int_{(0,1]} \frac{(g(t) \, \delta(t) - \delta(t)) \, (2 \, \delta(t) - g(t) \, \delta(t))}{t} \de \lambda(t)
  \\
  &   =  \int_{(0,1]} (g(t) -1) \, (2 - g(t)) \; \frac{\delta^2(t)}{t} \de \lambda(t)
      =  A \, \int_{(0,1]} h(g(t)) \de \mu(t)
  \\
  & \leq A \, h \left( \int_{(0,1]} g(t) \de \mu(t) \right)
      =  A \, h \left( \frac{1}{A} \, \int_{(0,1]} \delta^\prime(t) \, \delta(t) \de \lambda(t) \right)
      =  A \, h \left( \frac{1}{2A} \right)
  \\
  &   =  \frac{(4A-1) \, (1-2A)}{(4A-1) + 1}\,, 
\end{align*}
where the second last identity is due to integration by parts \cite{billingsley-1995}.
Therefore, 
\begin{align*}
  \xi(S_\delta) 
  &   =  \tau(S_\delta) - 2\,B
    \geq \tau(S_\delta) - 2 \; \frac{\tau(S_\delta) \, (1-2A)}{\tau(S_\delta) + 1}
      =  \frac{2 \, \tau(S_\delta)^2}{\tau(S_\delta) + 1}\,.
\end{align*}
It remains to show that the lower inequality is sharp as well.
Therefore, consider the copula diagonal $\delta_p$ discussed in Example \ref{ex:power}.
Straightforward calculation then yields
\begin{align*}
  \frac{2 \, \tau(S_\delta)^2}{1+\tau(S_\delta)}
  & = \frac{(2-p)^2}{p} = \xi(S_\delta)\,,
\end{align*}
where the last identity is due to Theorem \ref{Rep:Chatterjee.2}.
\end{proof}

\begin{figure}[h]
  \centering
  \includegraphics[width=0.45\textwidth]{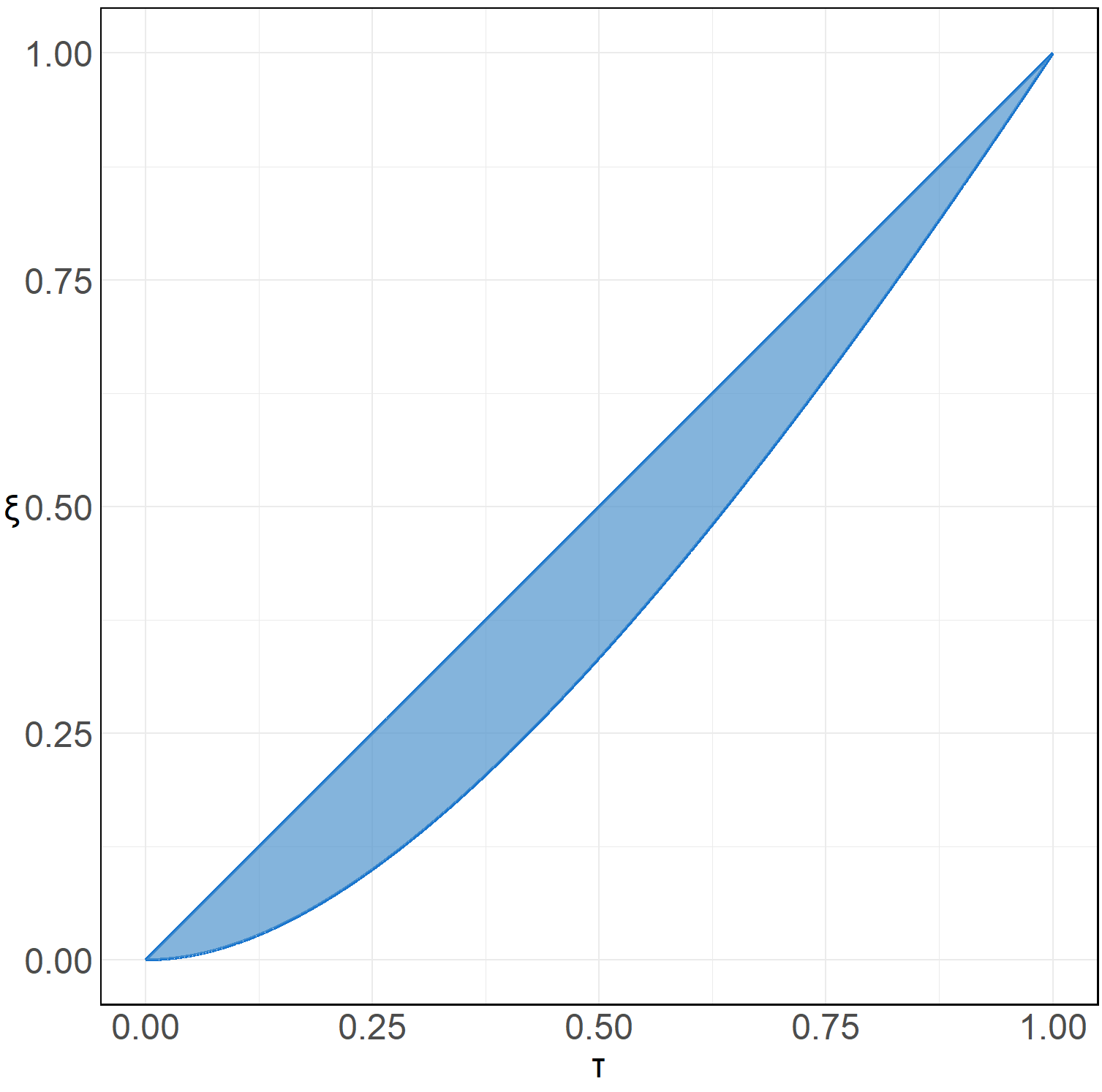}
\caption{The \(\tau\)-\(\xi\) region \( \Omega^{\rm LSL}_{\tau \text{,}\xi}\)}
\end{figure}

As mentioned in the proof of Theorem \ref{thm:markovtau}, the upper inequality is sharp for copula diagonals $l_a$ and $u_a$ and any choice of $a \in [0,1]$.
Theorem \ref{Rep:Chatterjee.2} offers an even more precise statement incorporating the shape constraints in \eqref{Def:LSL2}.

\begin{corollary} \label{thm:sharpineq}
Within the class of lower semilinear copula, the identity $\xi(S_\delta) = \tau(S_\delta)$ is attained if and only if, for $\lambda$-almost every $t \in [0,1]$, the copula diagonal $\delta$ fulfills either 
$t \, \delta^\prime(t) = \delta(t)$ or $t \, \delta^\prime(t) = 2 \, \delta(t)$.
\end{corollary}

Corollary \ref{thm:sharpineq} also applies to copula diagonals other than \(u_a\) and \(l_a\) such as those studied in Example \ref{ex:otherdiag}.
As stated in the proof of Theorem \ref{thm:markovtau}, the lower inequality is sharp for copula diagonals $\delta_p$ discussed in Example \ref{ex:power}. 
Therefore, according to Example  \ref{ex:MO}, the lower inequality is attained by Marshall-Olkin copulas $M_{\alpha,\alpha}$.

\subsection{Further inequalities for lower semilinear copulas}

We proceed by establishing connections between Chatterjee’s rank correlation and both Spearman’s rho and Spearman’s footrule.
The next result is a consequence of Theorem \ref{thm:markovtau} in combination with Theorem \ref{Thm:taurho}.

\begin{corollary} \label{cor:markovrho}
The inequality
\[
  \xi(S_\delta)\leq \rho(S_\delta)
\]
holds for all \(S_\delta \in \mathcal{C}^{\rm LSL}\) and it is sharp.
\end{corollary}

\begin{remark}
In \cite{ansari2025rockel}, the authors have established the inequality $\xi(C) \leq |\rho(C)|$
for copulas $C$ that are stochastically monotone, that is, the copula's partial derivate (where it exists) is monotone in the conditioning argument; see \cite{fuchs2023, fuchs2024} for more information on dependence properties of copulas.
It is important to note, however, that lower semilinear copulas typically fail to exhibit stochastic monotonicity; see \cite[Example 3.4]{maislinger2025}.
Consequently, the inequality stated in Corollary \ref{cor:markovrho} is not encompassed by the corresponding result in \cite{ansari2025rockel}.
\end{remark}

\begin{corollary}
Within the class of lower semilinear copulas, the identity \(\xi(S_\delta) = \rho(S_\delta)\)
is attained exclusively for the copulas of the form \(S_{l_a}\) with copula diagonal $l_a$.
\end{corollary}
\begin{proof}
The statement follows immediately from the fact that \(\tau(S_\delta)=\rho(S_\delta)\) if and only if \(\delta=l_a\) and that \(\xi(S_\delta)=\tau(S_\delta)\) whenever \(\delta=l_a\).
\end{proof}

The next result is a consequence of Theorem \ref{thm:markovtau} in combination with Theorem \ref{Thm:tauphi}.

\begin{corollary} \label{cor:markovphi}
The inequality
\[
  \xi(S_\delta) \leq \phi(S_\delta)
\]
holds for all \(S_\delta \in \mathcal{C}^{\rm LSL}\) and it is sharp.
\end{corollary}

\begin{corollary}
Within the class of lower semilinear copulas, the identity \(\xi(S_\delta) = \phi(S_\delta)\)
is attained exclusively for the copulas of the form \(S_{u_a}\) with copula diagonal $u_a$.
\end{corollary}
\begin{proof}
The statement follows immediately from the fact that \(\tau(S_\delta)=\phi(S_\delta)\) if and only if \(\delta=u_a\) (due to Corollary \ref{tauphi:sharp}) and that \(\xi(S_\delta)=\tau(S_\delta)\) whenever \(\delta=u_a\).
\end{proof}

For the sake of completeness, we list below the values of Chatterjee's rank correlation $\xi(S_\delta)$ for the diagonals introduced in earlier examples:
\begin{center}
    \begin{tabular}{l||cccc}
    \toprule
    $\delta$
    & $u_a$ (Ex.~\ref{ex:uala}) 
    & $l_a$ (Ex.~\ref{ex:uala})
    & $\delta_p$ (Ex.~\ref{ex:power}) 
    & $\delta_\alpha$ (Ex.~\ref{ex:frechet}) 
    \\
    \midrule
    $\xi(S_\delta)$ & $1-a^2$ & $a^4$ & $\frac{(2-p)^2}{p}$ & $\alpha^2$ \\
    \bottomrule
    \end{tabular}
\end{center}

Finally, we return to the discussion of Marshall–Olkin copulas initiated in Example \ref{ex:MO}.
While Marshall–Olkin copulas, in general, fail to be lower semilinear, the corresponding Markov product does exhibit lower semilinearity according to \cite{maislinger2025}.
For the calculation of $\xi(M_{\alpha,\beta})$ we can therefore make use of the connection between $\xi$ and $\phi$ stated in \eqref{Rep:Chatterjee}.

\begin{example}[Marshall-Olkin copulas] \label{ex:MO2}
According to \cite{fuchs2023JMVA}, the Markov product of a Marshall–Olkin copula is given by
\[
  M_{\alpha,\beta}^T\ast M_{\alpha,\beta}
  = \begin{cases}
      \Pi+\frac{\alpha^2}{1-2\alpha}\Pi \left(1-\Pi^{\frac{\beta-2\alpha\beta}{\alpha}}M^{\frac{2\alpha\beta-\beta}{\alpha}}\right) 
      & \alpha\notin \{0, \frac{1}{2}, 1\} 
      \\
      \Pi 
      & \alpha=0 
      \\
      \Pi+\frac{\beta}{2}\Pi(\log(M)-\log(\Pi)) 
      & \alpha=\frac{1}{2}
      \\
      M_{\beta,\beta} 
      & \alpha=1.
    \end{cases}
\]
Moreover, $M_{\alpha,\beta}^T\ast M_{\alpha,\beta} \in \mathcal{C}^{\rm LSL}$ due to \cite[Theorem 4.10]{maislinger2025}, and its corresponding diagonal satisfies
\begin{align*}
  \delta(t)
  & = 
  \begin{cases}
    t^2 \; \left(\frac{(1-\alpha)^2}{1-2\alpha} - \frac{\alpha^2}{1-2\alpha} t^{\beta \, \frac{1-2\,\alpha}{\alpha}} \right) & \alpha\notin \{0, \frac{1}{2}\}
    \\
    t^2 
    & \alpha=0 
    \\
    t^2 \left(1-\frac{\beta}{2}\log(t)\right) 
    & \alpha=\frac{1}{2}\,.
  \end{cases}
\end{align*}
Straightforward calculation incorporating the identity in \eqref{Rep:Chatterjee} then yields (see also \cite{ansari2024rockel})
\begin{align*}
  \xi(M_{\alpha,\beta})= 
  \begin{cases}
    0 & \text{if } (\alpha, \beta)=(0,0)
    \\
    \frac{2\alpha^2\beta}{3\alpha+\beta-2\alpha\beta} & \text{otherwise,}
\end{cases}
\end{align*}
noting that $3\alpha+\beta-2\alpha\beta = 0$ if and only if $(\alpha, \beta) = (0,0)$.
\end{example}

\begin{remark}
The established \(\tau\)-\(\xi\) region and the inequalities stated in Corollaries \ref{cor:markovrho} and \ref{cor:markovphi} are valid for lower semilinear copulas $S_\delta \in \mathcal{C}^{\rm LSL}$. 
However, for copulas whose copula diagonal $\delta$ belongs to the class $\mathcal{D}^{\rm LSL}$ but which are not themselves lower semilinear, the corresponding inequalities may not hold.
As a counterexample, consider the Marshall-Olkin copulas discussed in Example \ref{ex:MO2}: 
Although $\xi(M_{\alpha,\beta}) \leq \rho(M_{\alpha,\beta})$ for any choice of $(\alpha,\beta) \in [0,1]^2$ due to stochastic monotonicity of $M_{\alpha,\beta}$ in combination with \cite[Theorem 1.3]{ansari2025rockel}, the other inequalities can not be guaranteed since Marshall–Olkin copulas, in general, fail to be lower semilinear. As a counterexample, consider $(\alpha,\beta) = (1/2,3/4)$.
Then,  
$$ \frac{2 \, \tau(M_{\alpha,\beta})^2}{1+\tau(M_{\alpha,\beta})} = \frac{9}{35} > \frac{9}{36} = \xi(M_{\alpha,\beta})\,.$$
This contradicts Theorem \ref{thm:markovtau}.
\end{remark}

\subsection{Conclusion}

By proving the inequality $(1-\tau(S_\delta))^3 \leq (1-\rho(S_\delta))^2$ for lower semilinear copulas, we have confirmed the conjecture about the \(\tau\)-\(\rho\) region stated in \cite{maislinger2025}.
We extended the investigation to Spearman's footrule and Chatterjee's rank correlation and derived the exact \(\tau\)-\(\phi\), \(\phi\)-\(\rho\), and \(\tau\)-\(\xi\) regions for lower semilinear copulas.
Comparing the different regions, it is striking that the values of $\tau$ and $\phi$ differ only slightly, in contrast to the values of $\phi$ and $\rho$.
This is illustrated by computing, for each region, the proportion of the area in the unit square it occupies:
\begin{center}
\begin{tabular}{l||cccc}
\toprule
Pair & $(\tau\text{,} \rho)$ 
& $(\tau\text{,} \phi)$ 
& $(\phi\text{,} \rho)$
& $(\tau\text{,} \xi)$ 
\\
\midrule
Area 
& $\frac{1}{10} = 0.10 $ 
& $\frac{1}{14} \approx 0.07$ 
& $\frac{6}{35} \approx 0.17$
& $\frac{3}{2} - 2 \, \ln(2) \approx 0.11$
\\
\bottomrule
\end{tabular}
\end{center}
The idea of making the regions comparable by computing the proportion of each region's area has its origin in \cite{kokol2022}. 
We refer the reader to \cite{kokol2023,stopar2024,schreyer2017} for results on exact regions between concordance measures for the general class of copulas.

Although Chatterjee’s rank correlation is easily seen to be continuous for lower semilinear copulas with convex copula diagonals \cite[Theorem 3.5]{ansari2025Cont}, its continuity in the general case remains an open problem.
Therefore, convexity and compactness of the \(\tau\)-\(\xi\) region do not follow directly. We plan to tackle this problem in the near future.

Rather than limiting the analysis to pairwise comparisons of dependence measures, we see significant potential in extending the framework to higher-dimensional settings, such as triples. Pairwise results do not necessarily generalize to this case; for instance, simulations comparing $\tau$, $\rho$, and $\phi$ indicate that the feasible $\tau$-$\rho$-$\phi$ region for lower semilinear copulas is considerably smaller than the intersection of the corresponding pairwise regions; see Figure \ref{Fig:tauphirho} for an illustration.

\begin{figure}[hb!]
  \centering
  \includegraphics[width=0.475\textwidth]{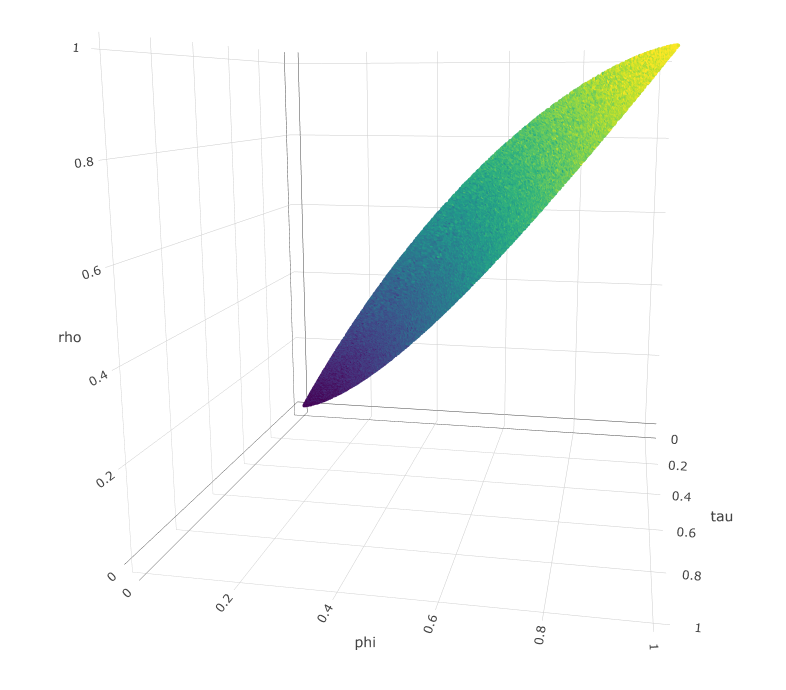}
  \includegraphics[width=0.475\textwidth]{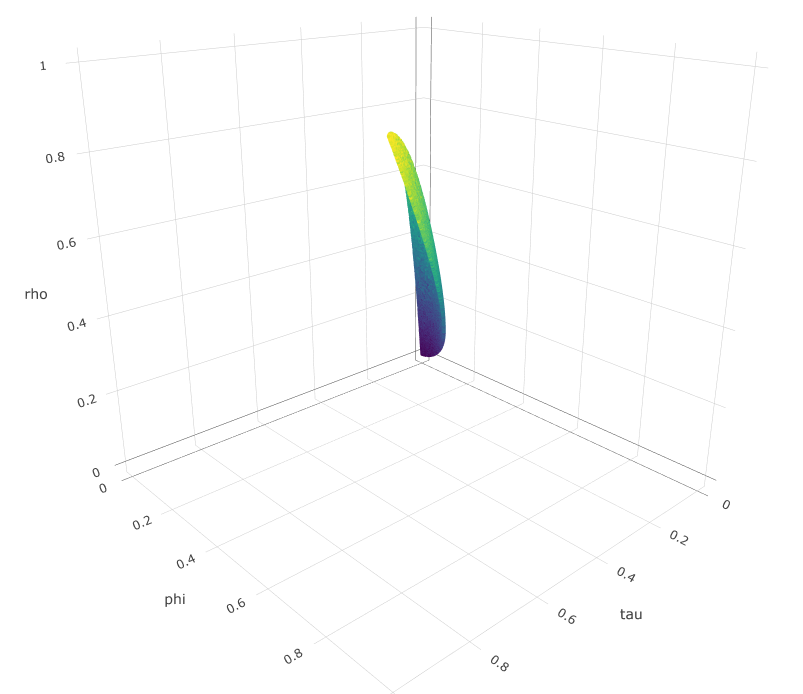}
\caption{Simulated \(\tau\)-\(\rho\)-\(\phi\) region from two different angles.}
\label{Fig:tauphirho}
\end{figure}

\section*{Acknowledgment}
The first and second authors gratefully acknowledge the support of the WISS 2025 project 'IDA-lab Salzburg' (20204-WISS/225/197-2019 and 20102-F1901166-KZP) and the support of the Austrian Science Fund (FWF) project {P 36155-N} \emph{ReDim: Quantifying Dependence via Dimension Reduction}.


\begin{thebibliography}{}

\bibitem[\protect\citeauthoryear{Ansari and Fuchs}{Ansari and
  Fuchs}{2025}]{ansari2025Cont}
Ansari, J. and S.~Fuchs (2025+).
\newblock On continuity of {C}hatterjee's rank correlation and related
  dependence measures.
\newblock {\em Available at \url{arxiv.org/abs/2503.11390}\/}.

\bibitem[\protect\citeauthoryear{Ansari and Rockel}{Ansari and
  Rockel}{2024}]{ansari2024rockel}
Ansari, J. and M.~Rockel (2024).
\newblock Dependence properties of bivariate copula families.
\newblock {\em Depend. Model.\/}, Article ID 20240002.

\bibitem[\protect\citeauthoryear{Ansari and Rockel}{Ansari and
  Rockel}{2025}]{ansari2025rockel}
Ansari, J. and M.~Rockel (2025+).
\newblock The exact region and an inequality between {C}hatterjee's and
  {S}pearman's rank correlations.
\newblock {\em Available at \url{arxiv.org/abs/2506.15897}\/}.

\bibitem[\protect\citeauthoryear{Billingsley}{Billingsley}{1995}]{billingsley-1995}
Billingsley, P. (1995).
\newblock {\em Probability and Measure\/} (3rd ed.).
\newblock John Wiley \& Sons, New York.

\bibitem[\protect\citeauthoryear{Chatterjee}{Chatterjee}{2020}]{chatterjee2020}
Chatterjee, S. (2020).
\newblock A new coefficient of correlation.
\newblock {\em J. Amer. Statist. Ass.\/}~{\em 116\/}(536), 2009--2022.

\bibitem[\protect\citeauthoryear{Dette, Siburg, and Stoimenov}{Dette
  et~al.}{2013}]{siburg2013}
Dette, H., K.~F. Siburg, and P.~A. Stoimenov (2013).
\newblock A copula-based non-parametric measure of regression dependence.
\newblock {\em Scand. J. Statist.\/}~{\em 40\/}(1), 21--41.

\bibitem[\protect\citeauthoryear{Durante}{Durante}{2006}]{durante2006}
Durante, F. (2006).
\newblock A new class of symmetric bivariate copulas.
\newblock {\em J. Nonparametr. Stat.\/}~{\em 18}, 499--510.

\bibitem[\protect\citeauthoryear{Durante, Kolesarova, Mesiar, and
  Sempi}{Durante et~al.}{2008}]{durante2008}
Durante, F., A.~Kolesarova, R.~Mesiar, and C.~Sempi (2008).
\newblock Semilinear copulas.
\newblock {\em Fuzzy Sets and Systems\/}~{\em 159}, 63--76.

\bibitem[\protect\citeauthoryear{Durante and Sempi}{Durante and
  Sempi}{2016}]{durante2016}
Durante, F. and C.~Sempi (2016).
\newblock {\em Principles of Copula Theory}.
\newblock CRC Press, Boca Raton FL.

\bibitem[\protect\citeauthoryear{Fredricks and Nelsen}{Fredricks and
  Nelsen}{2007}]{fredricks2007}
Fredricks, G.~A. and R.~B. Nelsen (2007).
\newblock On the relationship between {S}pearman's rho and {K}endall's tau for
  pairs of continuous random variables.
\newblock {\em J. Statist. Plann. Inference\/}~{\em 137\/}(7), 2143--2150.

\bibitem[\protect\citeauthoryear{Fuchs}{Fuchs}{2024}]{fuchs2023JMVA}
Fuchs, S. (2024).
\newblock Quantifying directed dependence via dimension reduction.
\newblock {\em J. Multivariate Anal.\/}~{\em 201}, Article ID 105266.

\bibitem[\protect\citeauthoryear{Fuchs and Limbach}{Fuchs and
  Limbach}{2025}]{fuchs2025DeMo}
Fuchs, S. and C.~Limbach (2025+).
\newblock A dimension reduction for extreme types of directed dependence.
\newblock {\em Available at \url{arxiv.org/abs/2506.04825}\/}.

\bibitem[\protect\citeauthoryear{Fuchs and Tschimpke}{Fuchs and
  Tschimpke}{2023}]{fuchs2023}
Fuchs, S. and M.~Tschimpke (2023).
\newblock Total positivity of copulas from a {Markov} kernel perspective.
\newblock {\em J. Math. Anal. Appl.\/}~{\em 518}, Article 126629.

\bibitem[\protect\citeauthoryear{Fuchs and Tschimpke}{Fuchs and
  Tschimpke}{2024}]{fuchs2024}
Fuchs, S. and M.~Tschimpke (2024).
\newblock A novel positive dependence property and its impact on a popular
  class of concordance measures.
\newblock {\em J. Multivariate Anal.\/}~{\em 200}, Article ID 105259.

\bibitem[\protect\citeauthoryear{Kokol~Bukov\v{s}ek and
  Moj\v{s}kerc}{Kokol~Bukov\v{s}ek and Moj\v{s}kerc}{2022}]{kokol2022}
Kokol~Bukov\v{s}ek, D. and B.~Moj\v{s}kerc (2022).
\newblock On the exact region determined by {S}pearman's footrule and {G}ini's
  gamma.
\newblock {\em J. Comput. Appl.Math\/}~{\em 410}, Article ID 114212.

\bibitem[\protect\citeauthoryear{Kokol~Bukov\v{s}ek and
  Stopar}{Kokol~Bukov\v{s}ek and Stopar}{2023}]{kokol2023}
Kokol~Bukov\v{s}ek, D. and N.~Stopar (2023).
\newblock On the exact regions determined by {K}endall's tau and other
  concordance measures.
\newblock {\em Mediterr. J. Math.\/}~{\em 20}, Article ID 147.

\bibitem[\protect\citeauthoryear{Kokol~Bukov\v{s}ek and
  Stopar}{Kokol~Bukov\v{s}ek and Stopar}{2024}]{stopar2024}
Kokol~Bukov\v{s}ek, D. and N.~Stopar (2024).
\newblock On the exact region determined by {S}pearman's rho and {S}pearman's
  footrule.
\newblock {\em J. Comput. Appl.Math\/}~{\em 437}, Article ID 115463.

\bibitem[\protect\citeauthoryear{Maislinger and Trutschnig}{Maislinger and
  Trutschnig}{2025}]{maislinger2025}
Maislinger, L. and W.~Trutschnig (2025).
\newblock On bivariate lower semilinear copulas and the star product.
\newblock {\em Int. J. Approx. Reason.\/}~{\em 179}, Article ID 109366.

\bibitem[\protect\citeauthoryear{Nelsen}{Nelsen}{2006}]{nelsen2006}
Nelsen, R.~B. (2006).
\newblock {\em An Introduction to Copulas.\/} (2nd ed.).
\newblock New York: Springer.

\bibitem[\protect\citeauthoryear{Pap}{Pap}{2002}]{pap2002}
Pap, E. (2002).
\newblock {\em Handbook of Measure Theory}.
\newblock North Holland, Amsterdam.

\bibitem[\protect\citeauthoryear{Schreyer, Paulin, and Trutschnig}{Schreyer
  et~al.}{2017}]{schreyer2017}
Schreyer, M., R.~Paulin, and W.~Trutschnig (2017).
\newblock On the exact region determined by {K}endall's $\tau$ and {S}pearman's
  $\rho$.
\newblock {\em J. R. Stat. Soc., Ser. B, Stat. Methodol.\/}~{\em 79}, 613--633.

\end{thebibliography}

\end{document}